\documentclass[letterpaper, 10 pt, conference]{ieeeconf}  % Comment this line out if you need a4paper

\IEEEoverridecommandlockouts                              % This command is only needed if 
\title{\LARGE \bf
Computationally Efficient Safe Reinforcement Learning for Power Systems
}

\author{Daniel Tabas and Baosen Zhang% <-this % stops a space
\thanks{This work is partially supported by the National Science Foundation Graduate Research Fellowship Program under Grant No. DGE-1762114 and NSF grant ECCS-1930605. Any opinions, findings, conclusions, or recommendations expressed in this material are those of the authors and do not necessarily reflect the views of the National Science Foundation.}% <-this % stops a space
\thanks{Authors are with the department of Electrical and Computer Engineering, University of Washington, Seattle, WA, United States.
        {\tt\small \{dtabas, zhangbao\}@uw.edu}.}%
}

\newcommand{\R}{\mathbb{R}}

\newcommand\m[1]{\begin{bmatrix}#1\end{bmatrix}}

\usepackage{amsthm,amsmath,amssymb,amsfonts,graphicx}
\usepackage{comment}
\usepackage{xcolor}

\newtheorem{theorem}{Theorem}
\newtheorem{lemma}{Lemma}
\usepackage[normalem]{ulem}

\theoremstyle{definition}
\newtheorem{definition}{Definition}

\newcommand{\X}{\textit{\textbf{X}}}
\newcommand{\D}{\textit{\textbf{D}}}
\newcommand{\U}{\textit{\textbf{U}}}
\renewcommand{\S}{\textit{\textbf{S}}}

\renewcommand{\P}{\textit{\textbf{P}}}
\newcommand{\Q}{\textit{\textbf{Q}}}
\renewcommand{\int}{\textbf{int}}
\newcommand{\bd}{\textbf{bd}}

\newcommand{\B}{\mathbb{B}}

\begin{document}

\maketitle
\thispagestyle{plain}
\pagestyle{plain}

%%%%%%%%%%%%%%%%%%%%%%%%%%%%%%%%%%%%%%%%%%%%%%%%%%%%%%%%%%%%%%%%%%%%%%%%%%%%%%%%
\begin{abstract}

%Large-scale adoption of renewable energy resources requires new control techniques that can optimally compensate for the variability of renewables while satisfying hard operational constraints. Data-driven approaches such as reinforcement learning (RL) provide promising ways to find control policies with good performance, but it is difficult to verify that these policies would be safe, that is, guaranteeing constraint satisfaction. 

%In this paper, we use techniques from set-theoretic control to craft a computationally efficient approach to safe RL. The proposed control policy consists of a neural network and a novel, closed-form safety filter that guarantees all actions chosen are safe. The construction of the closed-form safety filter exploits the properties of robust controlled-invariant sets represented by polytopes. The safety filter is differentiable, allowing end-to-end training of the policy using any off-the-shelf RL algorithm. We demonstrate the method on a frequency regulation problem in a 9-bus power system model and compare its performance favorably against two baseline control techniques.

We propose a computationally efficient approach to safe reinforcement learning (RL) for frequency regulation in power systems with high levels of variable renewable energy resources. The approach draws on set-theoretic control techniques to craft a neural network-based control policy that is guaranteed to satisfy safety-critical state constraints, without needing to solve a model predictive control or projection problem in real time. By exploiting the properties of robust controlled-invariant polytopes, we construct a novel, closed-form ``safety-filter'' that enables end-to-end safe learning using any policy gradient-based RL algorithm. We then apply the safety filter in conjunction with the deep deterministic policy gradient (DDPG) algorithm to regulate frequency in a modified 9-bus power system, and show that the learned policy is more cost-effective than robust linear feedback control techniques while maintaining the same safety guarantee. We also show that the proposed paradigm outperforms DDPG augmented with constraint violation penalties.

\end{abstract}

%%%%%%%%%%%%%%%%%%%%%%%%%%%%%%%%%%%%%%%%%%%%%%%%%%%%%%%%%%%%%%%%%%%%%%%%%%%%%%%%

\section{INTRODUCTION}

Power systems are a quintessential example of safety-critical infrastructure, in which the violation of operational constraints can lead to large blackouts with high economic and human cost. As variable renewable energy resources are integrated into the grid, it becomes increasingly important to ensure that the system states, such as generator frequencies and bus voltages, remain within a ``safe'' region defined by the operators \cite{chen2021}.

The design of safe controllers concerns the ability to ensure that an uncertain dynamical system will satisfy hard state and action constraints during execution of a control policy~\cite{mitchell2005time,lygeros2004reachability}. Recently, set-theoretic control \cite{Blanchini2015} has been applied to a wide range of safety-critical problems in power system operation \cite{El-Guindy2017,Zhang2020}. This approach involves computing a \emph{robust controlled-invariant set} (RCI) along with an associated control policy which is guaranteed to keep the system state inside the RCI \cite{Blanchini2015,El-Guindy2017a,El-Guindy2018}. If the RCI is contained in the feasible region of the (safety-critical) state constraints, then the associated control policy is considered to be safe.

However, the set-theoretic approach requires several simplifying assumptions for tractability, leading to controllers with suboptimal performance. First, the disturbances to the system are assumed to be bounded in magnitude but otherwise arbitrary \cite{Blanchini2015,Zhang2020}. Second, the RCIs must be restricted to simple geometric objects such as polytopes or ellipsoids \cite{Maidens2013}. Third, many approaches select an RCI and control policy in tandem, which usually requires the control policy to be linear and forces a tradeoff between performance and robustness \cite{Liu2019a,Blanchini1997,Nguyen1999}. Fourth, nonlinear systems must be treated as linear systems plus an unknown-but-bounded linearization error \cite{El-Guindy2017}.

Once an RCI is generated using the conservative assumptions listed above, data-driven approaches can use learning to improve performance with respect to the true behavior of the disturbances and nonlinearities without risk of taking unsafe actions \cite{Wabersich2021, Aswani2013,Gros2020,Cheng2019}. However, these techniques require solving a model predictive control (MPC) or projection problem each time an action is executed, which may be too computationally expensive. Several approaches that avoid repeatedly solving an optimization problem have also been proposed. One such approach involves tracking the vertices of the set of safe actions, and using a neural network to specify an action by choosing convex weights on these vertices. However, this is only possible when the RCI has exceedingly simple geometry \cite{Zheng2020}. Other strategies only guarantee safety in expectation, and do not rule out constraint violations in every situation \cite{pmlr-v70-achiam17a,Yu2019}. Controllers with Lyapunov stability or robust control guarantees have also been proposed~\cite{Cui2021,cui2021a,Donti2021}, but stability does not always translate to constraint satisfaction.

In this paper, we present a method to design safe, data-driven, and closed-form control policies for frequency regulation in power systems. Our approach combines the advantages of set-theoretic control and learning. In particular, we use simple linear controllers to find a maximal RCI, and then use reinforcement learning (RL) to train a neural network-based controller that improves performance while maintaining safety. The safety of this control policy is accomplished by constraining the output of the neural network to the present set of safe actions. By leveraging the structure of polytopic RCIs, we construct a closed-form \textit{safety filter} to map the neural network's output into the safe action set without solving an MPC or projection problem. The safety filter is differentiable, allowing end-to-end training of the neural network using any policy gradient-based RL algorithm. We demonstrate our proposed control design on a frequency regulation problem in a 9-bus power system model consisting of several generators, loads, and inverter-based resources (IBRs). The simulation results demonstrate that our proposed policy maintains safety and outperforms safe linear controllers without repeatedly solving an optimization problem in real time.

We focus on applying our algorithm to the problem of primary frequency control in power systems. Frequency is a signal in the grid that indicates the balance of supply and demand. Generators typically respond to the change in frequency by adjusting their power output to bring the frequency back to nominal (e.g., 60 Hz in the North American system)~\cite{Kundur1994,zhao2014design}. For conventional generators, these responses are limited to be linear (possibly with a dead-band). In contrast, IBRs such as solar, wind and battery storage can provide almost any desired response to frequency changes, subject to some actuation constraints~\cite{ademola2020frequency}. Currently, however, these resources still use linear responses, largely because of the difficulty in designing nonlinear control laws. Recently, RL based methods have been introduced in the literature~(see, e.g.~\cite{Chen2021rl} and the references within). However, most approaches treat safety and constraint satisfaction as soft penalties, and cannot provide any guarantees~\cite{Chen2021rl,yan2018data,latif2020state}.

% Give an outline of the paper.
The rest of the paper is organized as follows. Section \ref{sec:2} introduces the power system model and formulates the problem of safety-critical control from a set-theoretic perspective. Section \ref{sec:3} describes the proposed controller design. Section \ref{sec:4} presents simulation results for the modified 9-bus power system.
\section{MODEL AND PROBLEM FORMULATION} \label{sec:2}

\subsection{Model assumptions}

% Introduce model: $x_{t+1} = Ax_t + Bu_t + Ed_t, d_t \in \D$
In this paper we are interested in a linear system with control inputs and disturbances. We write the system evolution as \begin{equation} \label{eqn:linear_system}
x_{t+1} = Ax_t + Bu_t + Ed_t,
\end{equation}
where $x_t \in \R^n$, $u_t \in \R^m$ and $d_t \in \R^p$ are vectors of the state variables, control inputs, and disturbances at time $t$. We assume the disturbance $d_t$ is bounded but otherwise can take arbitrary values. More precisely, we assume that $d_t$ lies in a compact set. This boundedness assumption on $d_t$ is fairly general, since it allows the disturbances to capture uncontrolled input into the system, model uncertainties in $A$, $B$, and $E$, and linearization error. For more compact notation, we will sometimes summarize \eqref{eqn:linear_system} as $x^+= f(x_t,u_t,d_t)$.

% Write here that $u \in U$ and $d \in D$. Define polytopes in-line.
The constraints on inputs are $u_t \in \U \subset \R^m$ and $d_t \in \D \subset \R^p$ for all $t$. The sets $\U$ and $\D$ are assumed to be polytopes, defined as the bounded intersection of a finite number of halfspaces or linear inequalities \cite{Boyd2009}. Specifically, $\U$ and $\D$ are defined as \begin{align}
    \U &= \{u \in \R^m \mid -\bar{u} \leq V_u u \leq \bar{u}\} \text{ and} \label{eqn:10-7-3}\\
    \D &= \{d \in \R^p \mid -\bar{d} \leq V_d d \leq \bar{d}\}. \label{eqn:10-7-4}
\end{align}

% Describe state constraints for safety-critical control
In safety-critical control problems such as frequency regulation, operators want to keep the system states within hard constraints. For example, frequencies are generally kept within a tenth of a hertz of the nominal frequency and rotor angle deviations are limited for stability considerations~\cite{Kundur1994}. We use the set \begin{align}
    \X &= \{x \in \R^n \mid -\bar{x} \leq V_x x \leq \bar{x}\} 
\end{align} to denote the constraints that the state $x$ must satisfy in real-time. 

\subsection{Safety-critical control}
    % Motivate and define concept of RCI. RCI is defined in a definition environment.
    Because of the presence of disturbances, it may not be possible for the system state to always remain in $\X$. Some states close to the boundary of $\X$ could be pushed out by a disturbance no matter the control action, while for other states in $\X$, there may exist a control action such that no disturbance would push the state outside of the prescribed region. This motivates the definition of a \emph{robust controlled-invariant set}.
    
    \begin{definition}[Robust controlled-invariant set (RCI) \cite{Blanchini2015}]
    An RCI is a set $\S$ for which there exists a feedback control policy $u_t = \pi_0(x_t) \in \U$ ensuring that all system trajectories originating in $\S$ will remain in $\S$ for all time, under any disturbance sequence $d_t \in \D$.
    \end{definition}

    % Explain in paragraph form why and how we choose $(\S,\pi_0)$.
If $\S$ is contained in $\X$, then $\pi_0$ is a safe policy. Often, the goal is to find the policy that maximizes the size of $\S$ while being contained in $\X$, since it corresponds to making most of the acceptable states safe \cite{Maidens2013}.  In general, this is a difficult problem. Fortunately, if we restrict the policy to be linear, there are many well-studied techniques that have been shown to be successful at producing large safety sets \cite{Liu2019a,Liu2015,Tahir2010,Blanco2010}. 
%\todo{cite something}
    
    % However, the choice of $(\S,\pi_0)$ is not unique even under the restriction that $\S \subseteq \X$, but is subject to considerations for complexity and robustness \cite{Maidens2013}. 
    In this paper, we assume that $\S$ is a polytope described by $2r$ linear inequalities, and that $\pi_0$ is a linear feedback control policy. Specifically, we assume \begin{gather}
        \S = \{x \in \R^n \mid -\bar{s} \leq V_sx \leq \bar{s}\} \subseteq \X \text{ and} \label{eqn:10-7-1} \\
        \pi_0(x) = Kx \label{eqn:10-8-1}
    \end{gather} where $V_s \in \R^{r \times n}$, $\bar{s} \in \R^r$, and $Kx \in \U$ for all $x \in \S.$ For robustness, we choose the largest RCI satisfying \eqref{eqn:10-7-1}. The algorithm used for choosing $(\S,\pi_0)$ is described in \cite{Liu2015}. The algorithm uses a convex relaxation to find an \emph{approximately} maximal RCI $\S$ and an associated $K$ as the solution to an SDP. The objective of the SDP is to maximize the volume of the largest inscribed ellipsoid inside $\S$. 
    %\todo{Give one or two sentence about the algorithm and its performance.}
    
    % Define safe action set $\Omega(x)$ in paragraph form. Give $\Omega(x)$ its own equation with a number. Explain the difficulty of constraining an action to $\Omega(x)$.
    Of course, a linear policy that maximizes the size of $\S$ may not lead to satisfactory control performance. The set $\S$ is chosen jointly with the policy $\pi_0$, but there could be many policies (not necessarily linear) that keep $\S$ robustly invariant. We want to optimize over this class of nonlinear policies to improve the performance of the system.
    % Not only because of the additional flexibility afforded in being able to choose a potentially nonlinear controller, but also because $K$ is chosen to maximize $\S$ without consideration for optimizing the performance of the system.
    % However, note that while $\S$ is depends on the policy $\pi_0$, but there could may policies that achieves the same safety set $\S$.
    %  a fixed RCI $\S$, there may be many safe policies that have better performance than $Kx$: not only because of the additional flexibility afforded in being able to choose a nonlinear controller, but also because $K$ is chosen to maximize $\S$ without consideration for optimizing the performance of the system. 
     To explore the full range of safe policies, we define the \emph{safe action set at time $t$} as \begin{align}
        \Omega(x_t) := \{u_t \in \U \mid x_{t+1} \in \S, \ \forall\ d_t \in \D\} \label{eqn:10-8-2}
    \end{align} where it is assumed that $x_t \in \S$. By induction, any policy that chooses actions from $\Omega(x_t)$ is a safe policy \cite{Blanchini2015}. 
    
    %To find a better policy, we optimize over the set of safe policies: 
    We define the set of safe policies with respect to the RCI $\S$ as \begin{align}\Pi := \{\pi: \R^n \rightarrow \R^m \mid \pi(x_t) \in \Omega(x_t), \ \forall\ x_t \in \S\}. \end{align} Given $\S$, we search for a policy by optimizing over $\Pi$: 
    \begin{subequations} \label{eqn:opt_performance}
    \begin{gather}
        \min_{\pi \in \Pi} \mathbb{E}_{x_0 \in \S, d_t \in \D} \bigg[ \frac{1}{T} \sum_{t = 1}^T J(x_t,u_t) \bigg] \label{eqn:9-30-1}\\
        \text{subject to: } %x_0 \in \S \\
        x_{t+1} = Ax_t + Bu_t + Ed_t\\
        u_t = \pi(x_t)
    \end{gather}
    \end{subequations} where $\mathbb{E}_{x_0 \in \S, d_t \in \D}$ is the expectation with respect to randomness in initial conditions and in the sequence of disturbances, and $J(x_t,u_t)$ is the cost associated with occupying state $x_t$ and taking action $u_t$. To estimate \eqref{eqn:9-30-1}, $d_t$ is sampled from $\D$ but treated as stochastic, so that standard RL algorithms can be used to solve \eqref{eqn:opt_performance} \cite{Wabersich2021}. This relaxation does not require thorough sampling of $\D$ to preserve safety, since the constraint $\pi \in \Pi$ imposes state and input constraint satisfaction for \emph{all} possible disturbances $d_t \in \D$. The solution of \eqref{eqn:opt_performance} depends on the distribution of $d_t$ over $\D$ but safety is guaranteed for any $\pi \in \Pi$.
    % We assume that $x_0$ is uniformly distributed over $\S$.
    
   One example of a cost function that can be used in~\eqref{eqn:opt_performance} is the classical LQR cost on state and control~\cite{hespanha2018linear}, but other non-quadratic cost functions can also be used. For example, for sparsity-promoting controllers, we may set $J(x_t,u_t) = x_t^T Q x_t + c\|u_t\|_1,$ where $Q \succeq 0$ and $c>~ 0$~\cite{dorfler2014sparsity}.  
%   We are not constrained by classical cost functions when solving \eqref{eqn:opt_performance}.

\section{CONTROLLER DESIGN} \label{sec:3}
   \begin{figure}[ht]
        \centering
        \includegraphics[width=8cm]{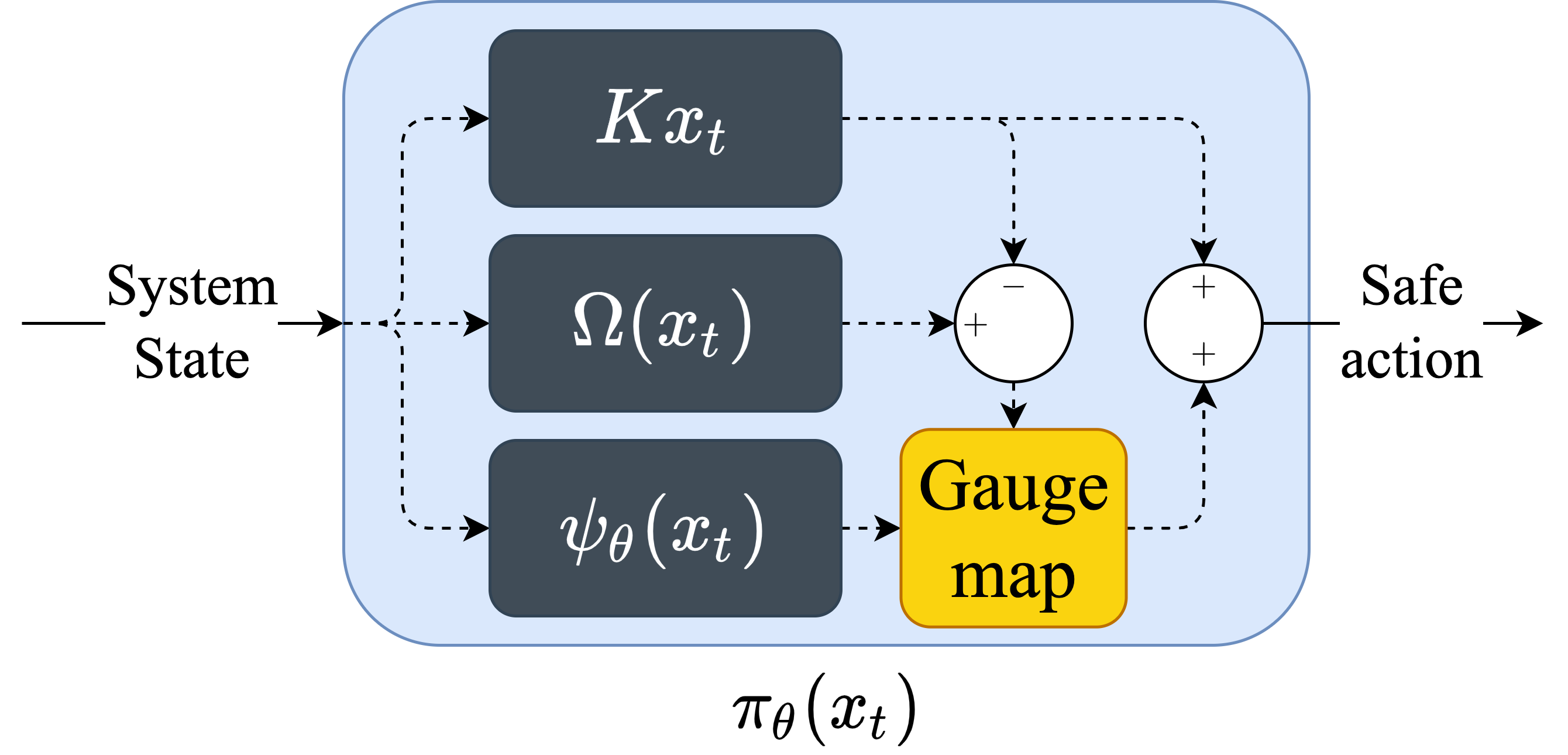}
        \caption{Policy network architecture for safe learning. The components are: $Kx_t$, a safe linear feedback included for numerical stability; $\Omega(x_t),$ the set of safe actions from observed state $x_t$; $\psi_\theta(x_t),$ a neural network; and the closed-form \textit{gauge map} which maps neural network outputs to the current set of safe actions $\Omega(x_t)$.}
        \label{fig:policy_network}
    \end{figure}
    
%\todo{Say a few sentences about using RL to find $\pi$, and how it is parametrized by a neural network.} 

In this section, we describe how set-theoretic control techniques can be used to create a safety guarantee for data-driven controllers without solving an MPC or projection problem in real time. Since $d_t$ is unknown, data-driven approaches for choosing $\pi$ are appropriate if safety guarantees can be maintained. For control problems with continuous state and action spaces, one class of RL algorithms involves parameterizing $\pi$ as a neural network or other function approximator and using stochastic optimization to search over the parameters of that function class for a (locally) optimal policy. 

A common approach to safety-critical control with RL is to combine a model-predictive controller with a neural network providing an action recommendation or warm start \cite{Wabersich2021,Aswani2013}. However, this makes it difficult to search over $\Pi$ efficiently and leads to control policies with higher computational overhead. One optimization-free approach involves tracking the vertices of $\Omega(x_t)$ and using a neural network to choose convex weights on the vertices of $\Omega(x_t)$. However, this is only possible when $\S$ has exceedingly simple geometry \cite{Zheng2020}.
%In this section, we provide an efficient parameterization of $\Pi$ that addresses these shortcomings.
While it is difficult to constrain the output of a neural network to arbitrary polytopes such as $\Omega(x_t)$, it is easy to constrain the output to $\B_\infty$, the $\infty$-norm unit ball in $\R^m$, using activation functions like sigmoid or hyperbolic tangent in the output layer. By establishing a correspondence between points in $\B_\infty$ and points in $\Omega(x_t),$ we will use neural network-based controllers to parameterize $\Pi$.
%\todo{Add in this paragraph or the next that the output of a neural network can be easily mapped into the infinity norm ball, how this is done, and say it's hard to map into other sets.}

%\todo{Talk about Fig.~\ref{fig:policy_network} in this paragraph.}
In particular, we construct a class of safe, differentiable, and closed-form policies $\pi_\theta$, parameterized by $\theta$, that can approximate any policy in $\Pi$. The policy first chooses a ``virtual'' action in $\B_\infty$ using a neural network $\psi_\theta$. The policy then uses a novel, closed-form, differentiable ``safety filter'' to equate $\psi_\theta(x_t)$ with an action in $\Omega(x_t)$. Figure \ref{fig:policy_network} illustrates the way $\psi_\theta$, $\Omega$, and $\pi_0$ are interconnected using a novel \emph{gauge map} in order to form the policy $\pi_\theta$. In order to efficiently map between $\B_\infty$ and $\Omega(x_t)$, we now introduce the concepts of \emph{C-sets} and \emph{gauge functions}.

% Define C-sets in a definition environment.

\begin{definition}[C-set \cite{Blanchini2015}] \label{def:c-set} A \emph{C-set} is a set that is convex and compact and that contains the origin as an interior point.
\end{definition}

Any C-set can be used as a ``measuring stick'' in a way that generalizes the notion of a vector norm \cite{Blanchini2015}. In particular, the gauge function (or Minkowski function) of a vector $v \in \R^m$ with respect to a C-set $\Q \subset \R^m$ is given by \begin{align} \gamma_\Q(v) = \inf\{\lambda \geq 0 \mid v \in \lambda \Q\}. \end{align} If $\Q$ is a polytopic C-set defined by $\{w \in \R^m \mid F_i^Tw \leq g_i, i = 1,\ldots,r\},$ then the gauge function is given by \begin{align}
    \gamma_\Q(v) = \max_i\Big \{\frac{F_i^Tv}{g_i}\Big \}, \label{eqn:10-7-2}
\end{align} which is easy to compute since it is simply the maximum over $r$ elements.
%\todo{need to explicitly point out that \eqref{eqn:10-7-2} is easy to compute.} 
Equation \eqref{eqn:10-7-2} is derived in Appendix \ref{app:gauge}.
%\todo{Say where exactly in the appendix.}
We will use \eqref{eqn:10-7-2} to construct a closed-form, differentiable bijection between $\B_\infty$ and $\Omega(x_t)$.

%The safety filter is constructed as a bijection between two C-sets. However, there is no guarantee that $0 \in \int(\Omega(x_t))$.

%In the next section, we will first construct a bijection between two C-sets and then explain how to find $u^* \in \int(\Omega(x_t))$, such that $\Omega(x_t) - u^*$ is a C-set.

% Define safety filter, but don’t use it as a common knowledge term. Link back to other papers. Reinforce the fact that this safety filter is closed-form.

\subsection{Gauge map}

% Whenever the two sets are polytopic C-sets, it is easy to map between them. Define the gauge function in-line and provide the gauge function of a polytope. Define the gauge map in paragraph form (with figure \ref{fig:c_set_map}). Give the gauge map its own line and equation number. Say that the gauge function is also called the Minkowski function.

We will first show how to use the gauge function to construct a bijection from $\B_\infty$ to any C-set $\Q$, and will then generalize to the case when $\Q$ does not contain the origin as an interior point. For any $v \in \B_\infty$, we define the \emph{gauge map from $\B_\infty$ to $\Q$} as \begin{align}
    G(v|\Q) = \frac{\|v\|_\infty}{\gamma_\Q(v)} \cdot v.
\end{align}

\begin{lemma} \label{lemma:bijection}
For any C-set $\Q$, the gauge map $G: \B_\infty \rightarrow \Q$ is a bijection. Specifically, $w = G(v|\Q)$ if and only if $w$ and $v$ have the same direction and $\gamma_\Q(w) = \|v\|_\infty.$
\end{lemma}

The proof of Lemma \ref{lemma:bijection} is provided in Appendix \ref{app:bijection}. By Lemma \ref{lemma:bijection}, choosing a point in $\B_\infty$ is equivalent to choosing a point in $\Q$. The action of the gauge map is illustrated in Figure \ref{fig:c_set_map}. 

We cannot directly use the gauge map to convert between points in $\B_\infty$ and points in $\Omega(x_t)$, since $\Omega(x_t)$ may not contain the origin as an interior point. Instead, we must temporarily ``shift'' $\Omega(x_t)$ by one of its interior points, making it a C-set. Lemma \ref{lemma:intpoint} provides an efficient way to achieve this.

    \begin{figure}
        \centering
        \includegraphics[width=8cm]{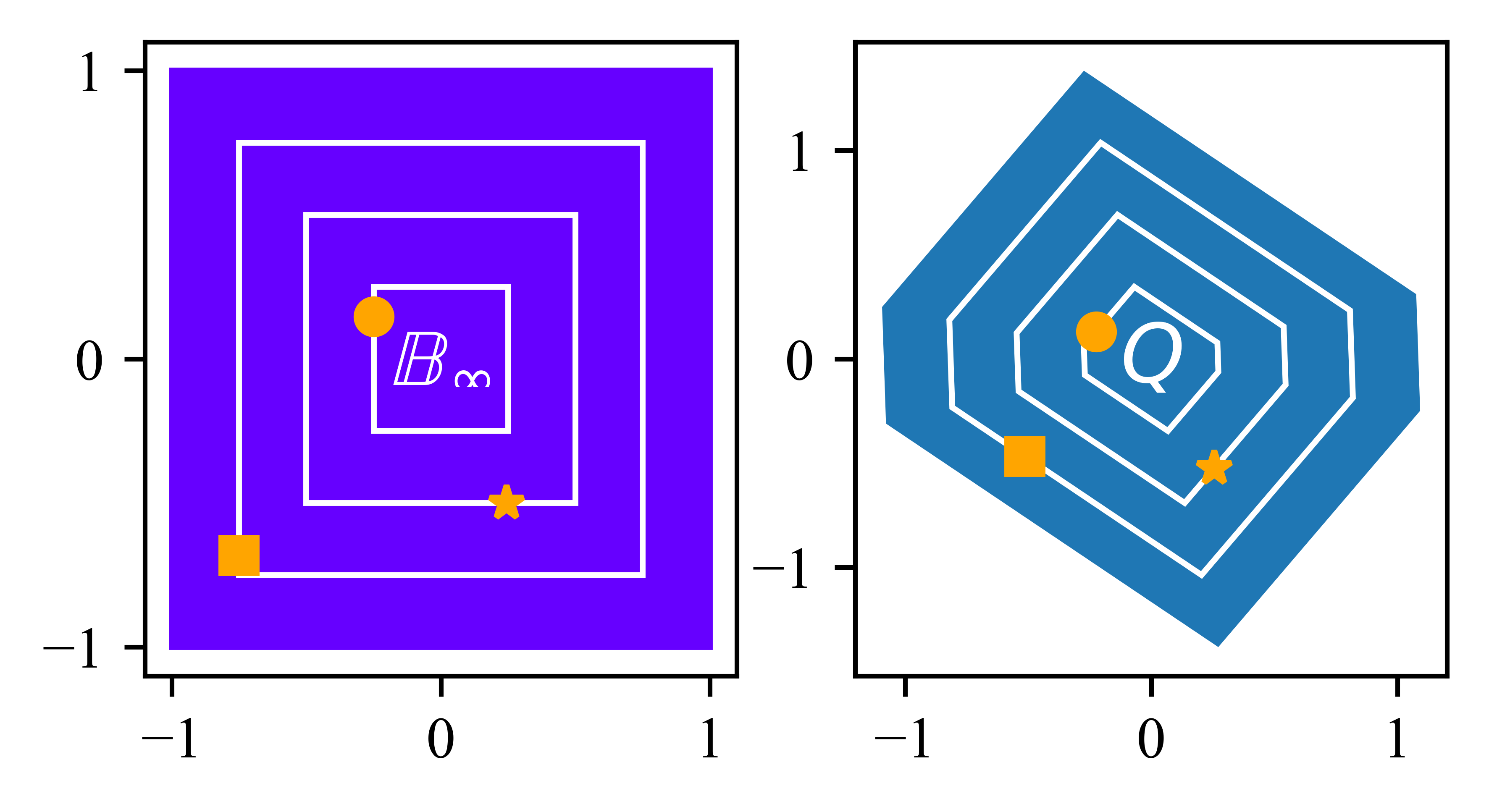}
        \caption{Action of the gauge map from $\B_\infty$ to randomly generated $\Q$, with the $\frac{1}{4}, \frac{1}{2},$ and $\frac{3}{4}$ level sets of the respective gauge functions shown in white. For each point $v \in \B_\infty$ and its image $w \in \Q$, $v$ and $w$ have the same direction and $\gamma_\Q(w) = \|v\|_\infty$.}
        \label{fig:c_set_map}
    \end{figure}

\begin{lemma} \label{lemma:intpoint}
If $\pi_0(x) = Kx$ is a policy in $\Pi$, then for any $x_t$ in the interior of $\S$, the set $\hat{\Omega}_t := [\Omega(x_t) - Kx_t]$ is a C-set.
\end{lemma}

The proof of Lemma \ref{lemma:intpoint} is provided in Appendix \ref{app:intpoint}. Figure \ref{fig:safety_filter} illustrates the way the gauge map and Lemma \ref{lemma:intpoint} are used in the policy network as a safety filter, by transforming the output of the policy network from $\B_\infty$ to $\Omega(x_t)$. 

\subsection{Policy architecture}

\begin{theorem} \label{thm:main}
Assume the system dynamics and constraints are given by \eqref{eqn:linear_system}, \eqref{eqn:10-7-3} and \eqref{eqn:10-7-4}, and there exists a choice of $(\S,\pi_0)$ conforming to \eqref{eqn:10-7-1} and \eqref{eqn:10-8-1}. Let $\psi_\theta: \S \rightarrow \B_\infty$ be a neural network parameterized by $\theta$. Then for any $x_t$ in the interior of $\S$, the policy
    \begin{align}
        \pi_\theta(x_t) := G(\psi_\theta(x_t) | \hat{\Omega}_t) + Kx_t \label{eqn:10-6-1}
    \end{align} has the following properties. \begin{enumerate}
        \item $\pi_\theta$ is a safe policy.
        \item $\pi_\theta$ can be computed in closed form.
        \item $\pi_\theta$ is differentiable at $x_t$. \label{thm:diff}
        \item $\pi_\theta$ can approximate any policy in $\Pi$.
    \end{enumerate}
\end{theorem} We will comment briefly on the last property and leave the proof of Theorem \ref{thm:main} to Appendix \ref{app:thm}. 
The ability of $\pi_\theta$ to approximate any policy in $\Pi$ given proper choice of $\theta$ is based on the function approximation properties of $\psi_\theta$ \cite{Hornik1989} and the ability of the gauge map to establish a one-to-one correspondence between points in $\B_\infty$ and actions in $\Omega(x_t)$. %\todo{say where in Appendix.} 
% Figure \ref{fig:policy_network} illustrates policy network architecture, while Figure \ref{fig:safety_filter} illustrates the way the gauge map is used in the policy network as a safety filter. 

    \begin{figure}
        \centering
        \includegraphics[width=8cm]{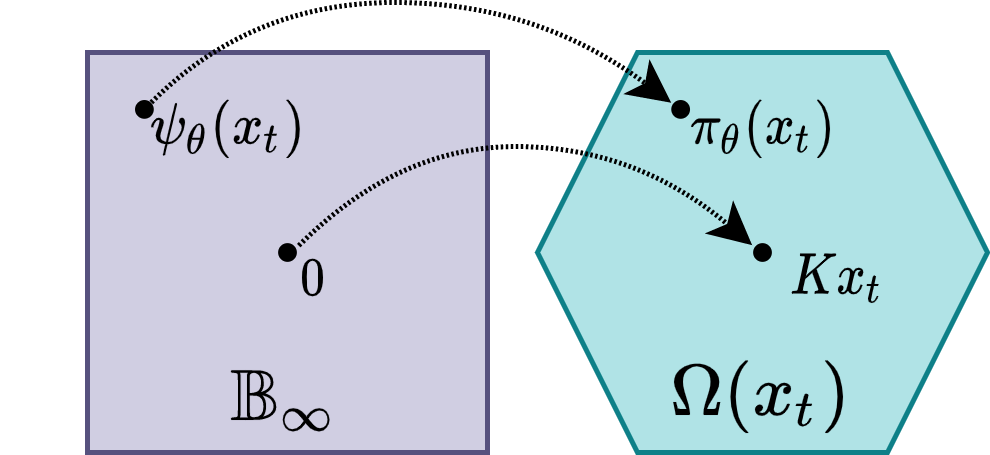}
        \caption{In the policy network, the gauge map is used to map virtual actions to safe actions.}
        \label{fig:safety_filter}
    \end{figure}

\subsection{Policy optimization through reinforcement learning}

% Explain how and why we choose $\theta$ using RL.
%Using \eqref{eqn:10-6-1}, we parametrize the search over $\Pi$ using a neural network with parameter $\theta$.
We parametrize the search over $\Pi$ using \eqref{eqn:10-6-1} with parameter $\theta$, and we choose $\theta$ to optimize \eqref{eqn:opt_performance} using policy gradient RL algorithms.
%Using the structure of the policy in \eqref{eqn:10-6-1}, we parameterize $\Pi$ with parameter $\theta$, and optimize $\theta$ using policy gradient RL algorithms.
% searching over $\theta$ provides an efficient way to search over $\Pi$ for an optimal policy. Since $d_t$ is unknown, data-driven approaches for choosing $\theta$ are appropriate. 
%\todo{Do we use policy gradient? If so, just say it here.}
The policy gradient theorem from reinforcement learning allows one to use past experience to estimate the gradient of the cost function \eqref{eqn:9-30-1} with respect to $\theta$ \cite{Bertsekas2019}. This is a standard approach for RL in continuous control tasks \cite{Lillicrap2016}. 
% Explain why it is important that $\pi_\theta$ is (sub)differentiable. This is part of our strategy to search over $\Pi$.
Policy gradient methods require that it be possible to compute the gradient of $\pi_\theta$ with respect to $\theta$. More specifically, $G$ must be differentiable (Thm. \ref{thm:main}, part \ref{thm:diff}) or else the safety filter would have to be treated as an uncertain influence whose behavior must be estimated from data. The parameter $\theta$ is randomly initialized at the beginning of the policy gradient algorithm.
    
In addition to being differentiable, $\pi_\theta$ has two other noteworthy attributes. First, under the optimal choice of $\theta$, the controller $\pi_\theta$ performs no worse than $\pi_0$. This is because $\pi_0$ is a feasible solution to \eqref{eqn:opt_performance}, so the optimal solution to \eqref{eqn:opt_performance} can do no worse. Second, unlike projection-based methods \cite{Gros2020}, the structure of $\pi_\theta$ facilitates exploration of the interior of the safe action set. This is because smooth functions such as the sigmoid or hyperbolic tangent can be used as activation functions in the output layer of $\psi_\theta$ to constrain its output to $\B_\infty$. By tuning the steepness of the activation function, it is possible to bias the output of $\psi_\theta$ towards or away from the boundary of $\B_\infty$.
\section{SIMULATIONS} \label{sec:4}

\subsection{Power system model}

% Describe 9-bus test system (Figure \ref{fig:power_system}).

% Introduce frequency regulation with continuous-time swing equation model in a per-bus format
The main application considered in this paper is frequency control in power systems. We consider a system with $N$ synchronous electric generators. The standard linearized swing equation at generator $i$ is:
\begin{subequations} \label{eqn:swing}
\begin{align}
\dot{\delta}_i &= \omega_{i} \\
M_i\dot{\omega}_i &= -D_i \omega_i -\sum_{j = 1}^N K_{ij} (\delta_i-\delta_j)+\sum_{k=1}^m b_{ik} u_k - \sum_{l=1}^p e_{il} d_l,
\end{align}
\end{subequations}
where $\delta_i$ is the rotor angle, $\omega_i$ is the frequency deviation, and $M_i$ and $D_i$ are the inertia and damping coefficients of generator $i$. The coefficients $K_{ij}$, $b_{ik}$, and $e_{il}$ are based on generator and transmission line parameters taken from a modified IEEE 9-bus test case, and are computed by solving the DC power flow equations. Thus, the size of the coefficient measures the influence of each element on the dynamics of generator $i$. The quantity $u_k$ represents controller $k$, an IBR such as a battery energy storage system or wind turbine~\cite{Kroposki2017,xu2018optimal}, where the active power injections can be controlled in response to a change in system frequency. The feasible control set $\U \subset \R^m$ represents limits on power output for each of the $m$ IBRs. % Therefore, the controller is bounded by some actuation constraints, such that $u_k$ needs to be in the interval $[\underline{u}_k,\overline{u}_k]$, which includes $0$ as an interior point.

% Describe the disturbances in frequency control. Say how $d_t$ is estimated, and why it gets different treatment (set-valued vs. stochastic).
The disturbance $d_l$ captures the uncertainties both in load and in uncontrolled renewable resources. It is also possible to use $d$ to account for parameteric uncertainties, linearization error associated with the linearized swing equation dynamics, or error associated with the DC power flow approximation, by adding virtual disturbances at every bus in the system \cite{El-Guindy2017a,Bumby2008}. 
%\todo{cite something}
The disturbance set $\D \subset \R^p$ is conservatively estimated from the capacity of the $p$ uncontrolled elements \cite{Chen2011}.
%For the purposes of safety constraints, we make the conservative assumption that the disturbance $d_l$ lies in an interval $[\underline{d}_l,\overline{d}_l]$ containing $0$ as an interior point, estimated from the total capacity of the variable resources \cite{Chen2011}. For the purposes of computing the performance of a controller with respect to a given performance metric, we assume $d_t$ is stochastic and follows a distribution estimated from historical data \cite{Wabersich2021}.
%\todo{anything we can cite here?} 
%This two-fold treatment of disturbances permits the use of standard reinforcement learning algorithms to choose a control policy, without sacrificing safety.

% Give the dynamics and constraints (on $u$ and $d$) for the discrete-time system in block format.
Discretizing the continuous-time system in \eqref{eqn:swing} and assembling block components gives a system in the form of \eqref{eqn:linear_system}. Let $\delta$ and $\omega \in \R^N$ be vectors representing the rotor angles and frequency deviations of all generators in the system, and let the system state be represented by $x = \m{\delta & \omega}^T \in \R^n$ where $n = 2N.$ Using time step $\tau$, the discrete-time system matrices are given by 
\begin{gather*}
A= \m{I & \tau I \\ -\tau M^{-1}K & I - \tau M^{-1}D}, \\
B=\m{0 \\ M^{-1} \hat{B}}, E=\m{0 \\ M^{-1} \hat{E}}
\end{gather*} where $[M]_{ii} = M_i$, $[D]_{ii} = D_i$, $[K]_{ij} = K_{ij}$, $[\hat{B}]_{ik} = b_{ik}$, and $[\hat{E}]_{il} = e_{il}$. 

We simulate the proposed policy network architecture on a 9-bus power system consisting of three synchronous electric generators, three controllable IBRs, and three uncontrolled loads. 
% The DC power flow equations are solved in order to write the system dynamics in the form of \eqref{eqn:swing}. 
The time step for discretization is 0.05 seconds. The load is modeled as autoregressive noise defined by \begin{align}
    d_{t+1} = \alpha d_t + (1-\alpha) \hat{d}
\end{align} where $\hat{d} \in \R^p$ is randomly generated from a uniform distribution over $\D,$ and $\alpha \in (0,1).$ The system is illustrated in Figure \ref{fig:power_system}.
%\todo{Is $d$ a vector?} 
%\todo{On the graph, indicate where the disturbances are coming in.}

\begin{figure}
    \centering
    \includegraphics[width=7cm]{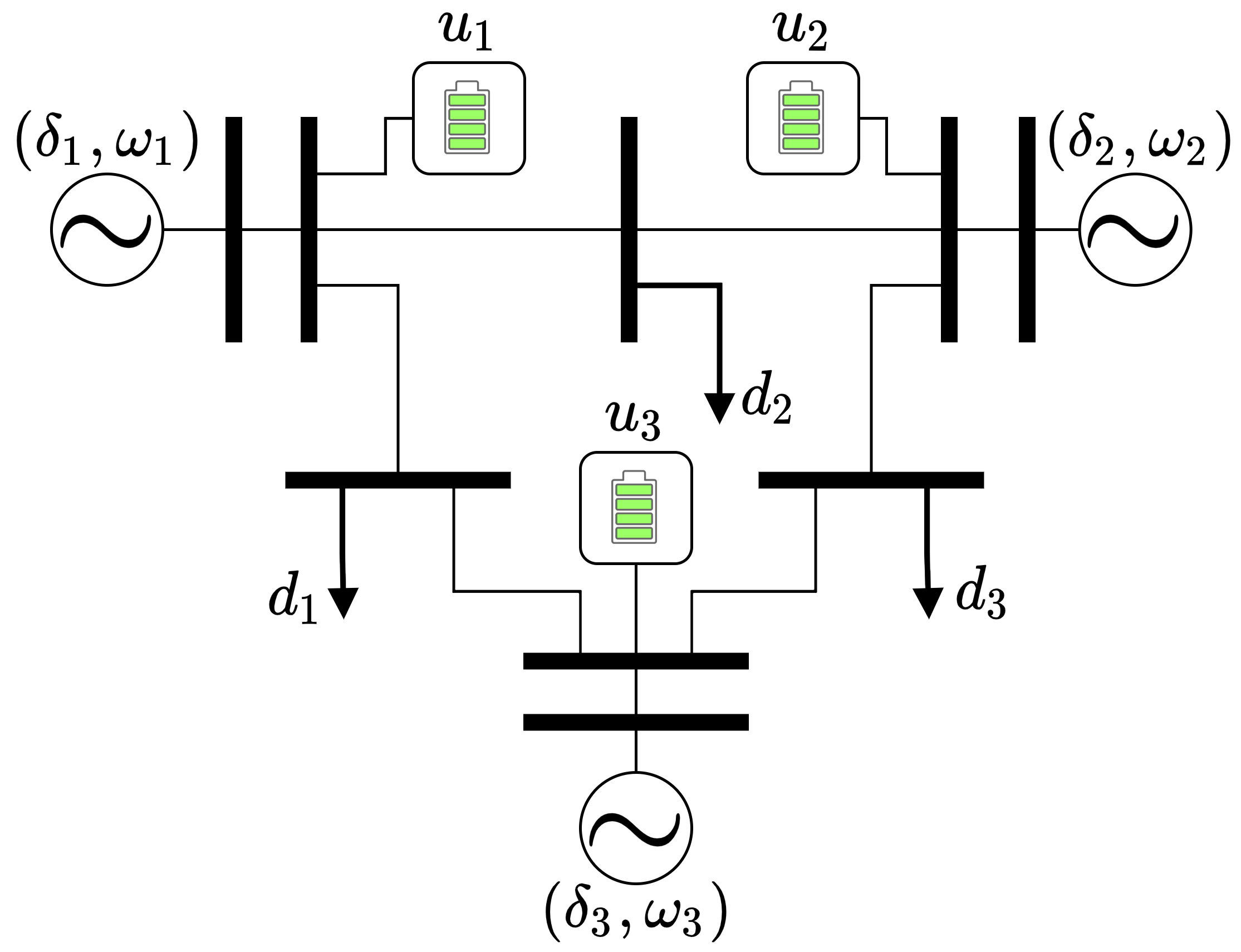}
    \caption{Illustration of 9-bus power system model.}
    \label{fig:power_system}
\end{figure}

\subsection{RL algorithm}

% Describe RL algorithm and parameters
To train the policy network, we used the Deep Deterministic Policy Gradient (DDPG) algorithm \cite{Lillicrap2016}, an algorithm well-suited for RL in continuous control tasks. DDPG is an actor-critic algorithm, in which the ``actor'' or policy chooses actions based on the state of the system, and the ``critic'' predicts the value of state-action pairs in order to estimate the gradient of the cost function \eqref{eqn:9-30-1} with respect to $\theta$ (the ``policy gradient''). In our simulations, the cost was given by \begin{align} J(x_t,u_t) = x^TQx + u^TRu \end{align} where $Q = \textbf{block diag}\{1000 I_N,10 I_N\}$, $R = 5 I_m$, and $I_N$ is the identity matrix in $\R^{N \times N}$. The actor was given by \eqref{eqn:10-6-1}. The function $\psi_\theta$ was parameterized by a neural network with two hidden layers of 256 nodes each, with ReLU activation functions in the hidden layers. We use sigmoid functions in the last layer to limit the the outputs to be within $[-1,1]$. The critic, or value network, had the same hidden layers as $\psi_\theta$ but a linear output layer. We trained the system for 200 episodes of 100 time steps each.

\subsection{Benchmark comparisons}

% Describe benchmarks (linear control; policy network with soft penalty on constraint violations)
To demonstrate the advantages of the proposed policy architecture, we compare against two benchmarks. The first is the linear controller $Kx$, chosen to maximize the size of the associated RCI. Using the algorithm in \cite{Liu2015}, we choose $(\S,K)$ by solving the optimization problem \begin{subequations} \label{eqn:9-29-1}
    \begin{align}
        \max_{\S \in \mathcal{S}, K \in \R^{n \times m}}& \text{vol}(\S) \\
        \text{s.t. Invariance: }&(A + BK)\S \oplus E\D \subseteq \S \\
        \text{Safety: }&\S \subseteq \X \\
        \text{Control bounds: }&K \S \subseteq \U
    \end{align} 
    \end{subequations} where $\oplus$ denotes Minkowski set addition and $\mathcal{S}$ is a class of polytopes described by \eqref{eqn:10-7-1}.
Figure \ref{fig:costs} displays the accumulated cost during a number of test trajectories, showing that $\pi_\theta$ is a more cost-effective controller than $Kx$ when the same $\S$ is used for each policy. This makes sense, since the nonlinear policy is afforded additional flexibility in balancing performance and robustness. Since $\pi_\theta$ and $Kx$ are both policies in $\Pi$, the learned policy has the same safety guarantees as the linear policy.
%\todo{Say what the cost function is somewhere}

\begin{figure}
    \centering
    \includegraphics[width=7cm]{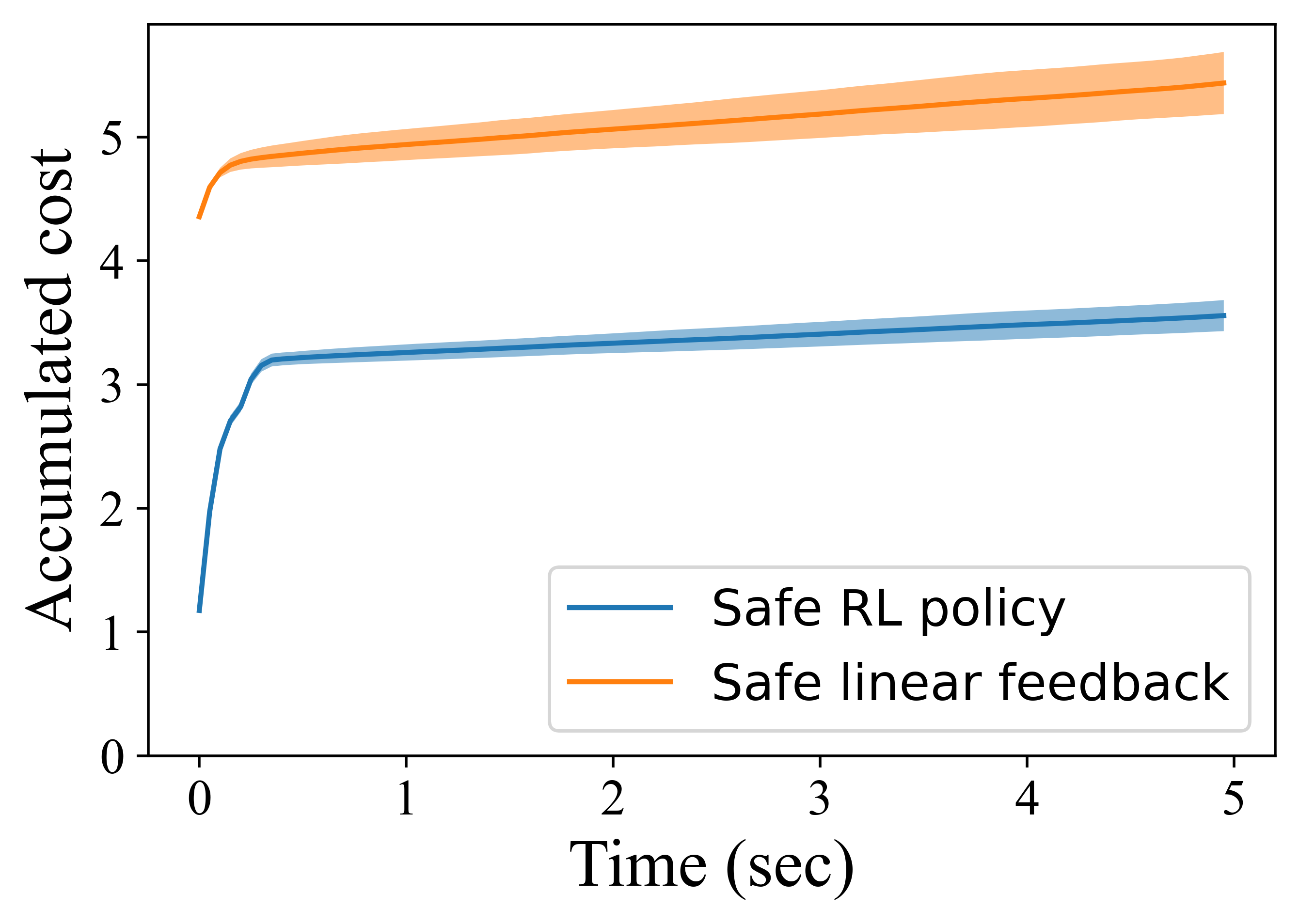}
    \caption{Accumulated cost over several test trajectories with a fixed initial condition and randomly generated disturbance sequences, showing that the RCI policy network achieves better performance than the safe linear feedback.}
    \label{fig:costs}
\end{figure}

The second benchmark is a policy network that is trained using DDPG augmented with a soft penalty on constraint violations, in order to incentivize remaining in $\X$. 
%\todo{Give more information for the baseline policy. Like the number of layers, what the soft penalties are...}
The policy network for this benchmark consists of two 256-node hidden layers with ReLU activation, and hyperbolic tangent activation functions in the output layer that clamp the output to the box-shaped set $\U$. The soft penalty is the total constraint violation, calculated as 
\begin{align*}
    \lambda \|\max\{V_x x_t - \bar{x},0\} - \min\{V_x x_t + \bar{x},0\}\|_1
\end{align*} where the $\max$ and $\min$ are taken elementwise, and $\lambda > 0$.

In Fig.~\ref{fig:max_angle_train}, we plot an example of the maximum angle deviation during training. We place a hard constraint of 0.1 radians on this angle deviation. For the policy given by \eqref{eqn:10-6-1}, the trajectory always stays within this bound, by the design of the the controller. For a policy trained with a soft penalty, trajectories initially exit the safe set. With enough training, the trajectories eventually satisfy the state constraints. 

Figure~\ref{fig:max_angle_test} shows that safety in training does not imply safety in testing. The policy network trained using a soft penalty can still result in constraint violations, whereas the safe policy network does not. In some sense, this is not unexpected. Only a limited number of disturbances can be seen during training, and because of the nonlinearity of the neural network-based policy, it is difficult to provide guarantees from the cost alone. In addition, picking the right soft penalty parameter is nontrivial. If the penalty $\lambda$ is too low, constraint satisfaction will not be incentivized, and if $\lambda$ is too high, convergence issues may arise \cite{Yoo2021}. In our experiments, we tuned $\lambda$ by hand to strike the middle ground, but even automatic, dynamic tuning of $\lambda$ during training is not guaranteed to prevent constraint violations in all cases \cite{Yoo2021}.

% Figure \ref{fig:max_angle_train} shows that even after training this benchmark policy until constraint violations are consistently avoided, violations are still possible (Figure \ref{fig:max_angle_test}). 
%\todo{Is Fig.~\ref{fig:max_angle_test} right? Is the limit 1 (where the dashed lines are)?} \comment{Yes, the limit is 0.1 for angle and 1 for frequency.}
% Soft penalties cannot provide rigorous safety guarantees: if the penalty $\lambda$ is too low, constraint satisfaction will not be incentivized, and if $\lambda$ is too high, convergence issues may arise \cite{Yoo2021}. Even if $\lambda$ is just right, constraint satisfaction will only be incentivized on average, and violations in rare circumstances will not be ruled out.
%\todo{Also, need to say more here. Emphasis why soft penalties can not provide guarantees.} 

% Compare safety performance (PN vs. RCI-PN) during training (Figure \ref{fig:max_angle_train}).

% Compare safety performance (PN vs. RCI-PN) during test (Figure \ref{fig:max_angle_test}).

% Compare costs (RCI-PN vs. linear) during testing (Figure \ref{fig:costs}). 

\begin{figure}
    \centering
    \includegraphics[width=7cm]{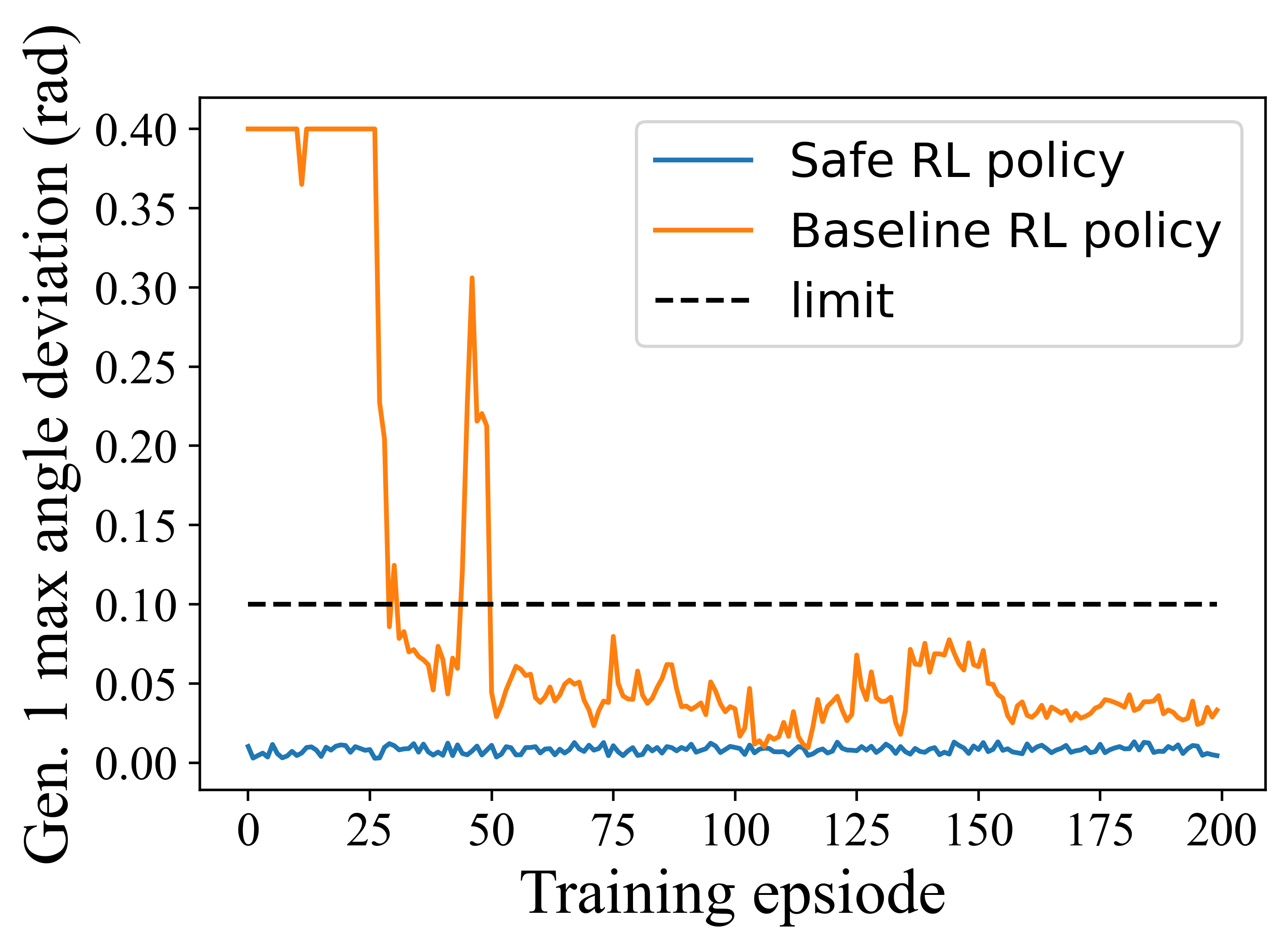}
    \caption{Maximum angle deviation per training episode for the safe policy network (blue) and the baseline policy network with soft penalty (orange). The safe policy network guarantees safety during training, while soft penalties eventually drive the baseline policy towards constraint satisfaction.}
    \label{fig:max_angle_train}
\end{figure}

\begin{figure}
    \centering
    \includegraphics[width=7cm]{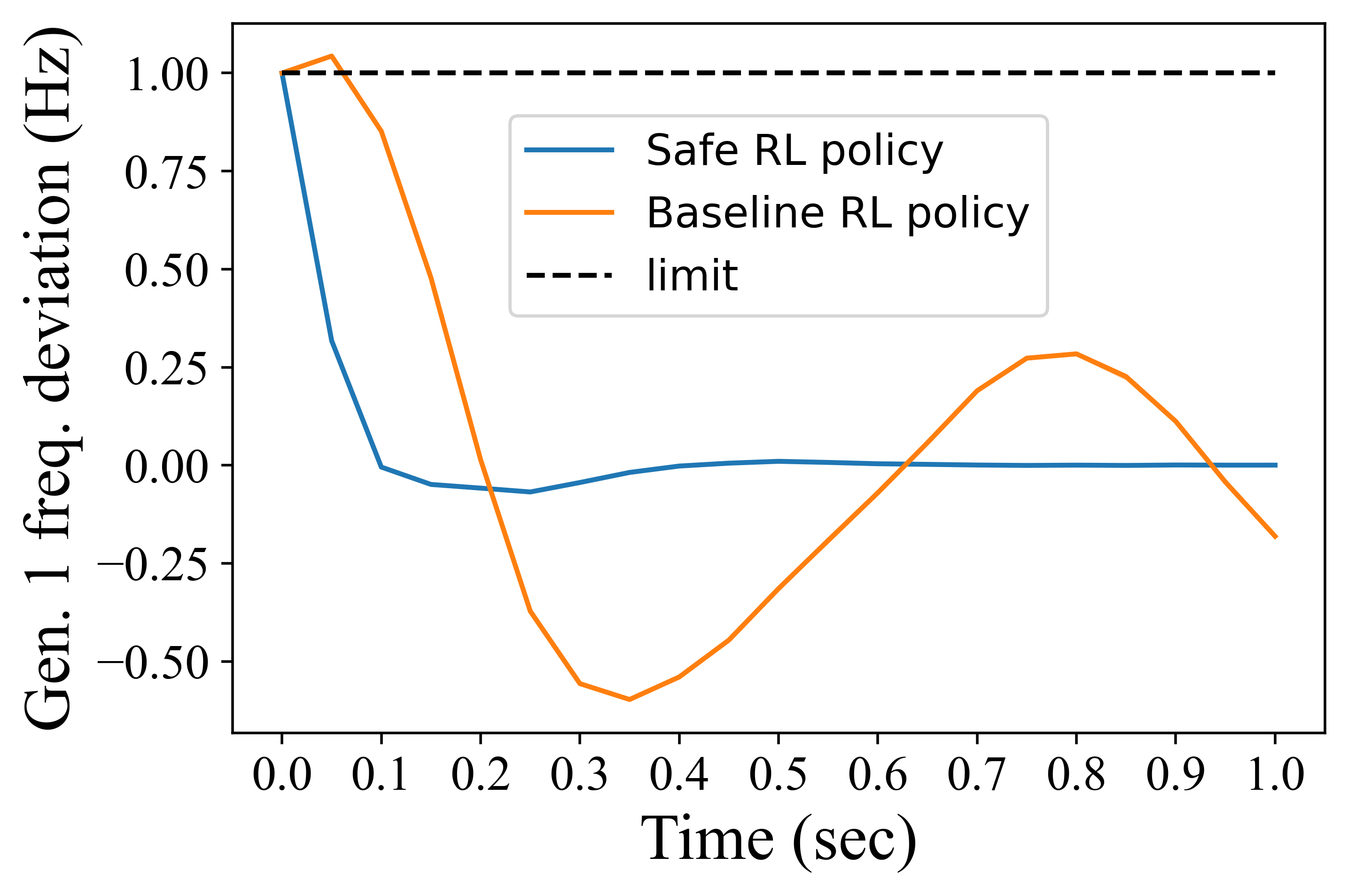}
    \caption{Even though soft penalties succeed in driving policies to be safe during training, they do not necessarily provide safety during testing. In this example, a policy network that was safe during training (Fig.~\ref{fig:max_angle_train}) still exhibits constraint violations during testing. In contrast, starting from the same initial conditions and subject to the same disturbance sequence, the proposed safe policy network guarantees constraint satisfaction.} %\todo{remove the dashed line at the top, put the dashed line where the constraints are.}}
    \label{fig:max_angle_test}
\end{figure}

\section{CONCLUSIONS}

In this paper, we propose an efficient approach to safety-critical, data-driven control. The strategy relies on results from set-theoretic control and convex analysis to provide provable guarantees of constraint satisfaction. Importantly, the proposed policy chooses actions without solving an optimization problem, opening the door to safety-critical control in applications in which computational power may be a bottleneck. We apply the proposed controller to a frequency regulation problem in power systems, but the applications are much more wide-ranging. Future work includes investigating robustness to changes in power system topology and extending the proposed technique to decentralized control.

\addtolength{\textheight}{-3cm}   % This command serves to balance the column lengths

\section*{ACKNOWLEDGMENTS}

D. Tabas would like to thank Liyuan Zheng for guidance and Sarah H.Q. Li for helpful discussions. 

\bibliography{references}
\bibliographystyle{ieeetr}

\appendix

\addtolength{\textheight}{-2cm}   % This command serves to balance the column lengths

% Derivation of the gauge function for polytopes
\subsection{Derivation of \eqref{eqn:10-7-2}} \label{app:gauge}
Let $\Q = \{x \in \R^n \mid Fx \leq g\}$ be a C-set, where $F \in \R^{r \times n},$ $g \in \R^r$, $F_i^T$ denotes the $i$th row of $F$, and $g_i$ denotes the $i$th element of $g$. The gauge function $\gamma_\Q(v)$ is computed as follows. \begin{align}
    \gamma_\Q(v) &= \inf\{\lambda \geq 0 \mid v \in \lambda \Q\}\\ 
    &= \inf\{\lambda \geq 0 \mid \frac{1}{\lambda} F_i^Tv \leq g_i, i = 1,\ldots, r\}\\
    &= \inf\{\lambda \geq 0 \mid \lambda \geq \frac{F_i^Tv}{g_i}, i = 1,\ldots,r\}\\
    &= \max\{0,\max_i\{\frac{F_i^Tv}{g_i}\}\}. \qedhere
\end{align} We now argue that $\max_i\{\frac{F_i^Tv}{g_i}\} \geq 0.$ If $F_i^Tv < 0 $ for all $i$, then $\Q$ is unbounded in the direction of $v$ and $\Q$ cannot be a C-set, a contradiction. Further, since $0 \in \int(\Q),$ it must hold that $g_i > 0$ for all $i$. Therefore, there exists $i$ such that~$\frac{F_i^Tv}{g_i} \geq 0.$
\hfill \qedsymbol

\subsection{Additional lemmas}
The following lemma will be used in the proofs of Lemma \ref{lemma:intpoint} and Theorem \ref{thm:main}. 

\begin{lemma} \label{lemma:polytope}
Under the assumptions of Theorem \ref{thm:main}, the safe action set $\Omega(x_t)$ is a polytope for all $x_t \in \S$.
\end{lemma}

\begin{proof}
Starting from \eqref{eqn:linear_system}, \eqref{eqn:10-7-1}, and \eqref{eqn:10-8-2}, the safe action set is \begin{alignat}{2}
    \Omega(x_t) &&= \{u_t \in \U \mid -\bar{s} & \leq V_s x_{t+1} \leq \bar{s}, \ \forall\ d_t \in \D\}\\
    &&= \{u_t \in \U \mid -\bar{s}_i & - \min_{d \in \D} V_s^{(i)T}Ed \nonumber \\ 
    && &\leq  V_s^{(i)T} (Ax_t + Bu_t) \nonumber \\
    && &\leq \bar{s}_i - \max_{d \in \D} V_s^{(i)T}Ed, \nonumber \\
    && & \hspace{1em} \forall\ i = 1,\ldots,r\} \label{eqn:10-8-3}
\end{alignat} where $\bar{s}_i$ is the $i$th element of $\bar{s}$ and $V_s^{(i)T}$ is the $i$th row of $V_s$. Since the $\min$ and $\max$ terms evaluate to constant scalars for each $i$, and since $x_t$ is fixed, \eqref{eqn:10-8-3} is a set of linear inequalities in $u_t$, making $\Omega(x_t)$ a polytope \cite{Blanchini2015}. \hfill \qedhere
\end{proof}

% Proof of lemma that C-set is a bijection
\subsection{Proof of Lemma \ref{lemma:bijection}} \label{app:bijection}
We will prove the more general case in which $\B_\infty$ is replaced by any polytopic C-set. Let $\P$ and $\Q$ be two polytopic C-sets, and define the gauge map from $\P$ to $\Q$ as $G(v|\P,\Q) = \frac{\gamma_\P(v)}{\gamma_\Q(v)} \cdot v$. We will prove that $G$ is a bijection from $\P$ to $\Q$. The proof is then completed by noting that $\gamma_{\B_\infty}$ is the same as the $\infty$-norm.

To prove injectivity, we fix $v_1,v_2 \in \P$ and show that if $G(v_1|\P,\Q) = G(v_2|\P,\Q)$ then $v_1 = v_2$. Assume $G(v_1|\P,\Q) = G(v_2|\P,\Q)$. Then $v_1$ and $v_2$ must be nonnegative scalar multiples of each other, i.e. $v_2 = \beta v_1$ for some $\beta \geq 0$. Making this substitution and applying positive homogeneity of the gauge function \cite{Blanchini2015} yields \begin{align}
    G(v_2|\P,\Q)
    = \frac{\gamma_\P(v_2)}{\gamma_\Q(v_2)}v_2 = \frac{\gamma_\P(v_1)}{ \gamma_\Q(v_1)}v_2.
\end{align} Noting that $G(v_1|\P,\Q)
= \frac{\gamma_\P(v_1)}{\gamma_\Q(v_1)}v_1$, we conclude that $\beta = 1$, thus $v_1 =~v_2$. 

To prove surjectivity, fix $w \in \Q$. We must find $v \in \P$ such that $G(v|\P,\Q) = w$. Since $\P$ and $\Q$ are C-sets, each set contains an open ball around the origin, thus $\P$ and $\Q$ each contain all directions at sufficiently small magnitude. Choose $v$ in the same direction as $w$ such that $\gamma_\P(v) = \gamma_\Q(w).$ Since $w \in \Q$, $v \in \P.$ Then, we have \begin{align}
    G(v|\P,\Q) &= \frac{\gamma_\P(v)}{\gamma_\Q(v)} v \\
    &= \frac{\gamma_\Q(w)}{\gamma_\Q(v)} v.
\end{align} Since $v$ and $w$ are in the same direction, $\frac{v}{\gamma_\Q(v)} = \frac{w}{\gamma_\Q(w)}$. Making this substitution completes the proof. \hfill \qedsymbol

% Proof of lemma that $\Omega(x_t) - \pi_0(x_t)$ is a C-set for all $x_t \in \int(\S)$
\subsection{Proof of Lemma \ref{lemma:intpoint}} \label{app:intpoint}
Let $\int$ and $\bd$ denote the interior and boundary of a set, and rewrite \eqref{eqn:10-8-3} as $\Omega(x) = \{u \in \R^m \mid Hu \leq h, F(Ax + Bu) \leq g \}$. Fix $x \in \int(\S)$ and let $u^* = Kx$. By Lemma \ref{lemma:polytope}, $\Omega(x)$ is convex and compact. To fulfill the properties of a C-set, it remains to show that $u^* \in \int(\Omega(x)).$ Since $\pi_0 \in \Pi$, $u^* \in \Omega(x)$. Assume for the sake of contradiction that $u^* \in \bd(\Omega(x))$. Then either $F_i^T(A+BK)x = g_i$ or $H_j^TKx = h_j$ for some $i$ or $j$, where the subscript denotes a row index. Suppose without loss of generality that the former holds, i.e. $F_i^T(A+BK)x = g_i$ for some $i$. Since $x \in \int(\S),$ there exists $\varepsilon \in (0,1)$ and $\alpha = [1 + \varepsilon \cdot \textbf{sign}(g_i)]$ such that $y = \alpha x$ is also in $\int(\S)$.
The set $\Omega(y)$ is contained in the halfspace $\{u \mid F_i^T(Ay + Bu) \leq g_i\}.$
Evaluating this inequality with $u = Ky$, we have $F_i^T(A + BK)y = \alpha F_i^T(A+BK)x = \alpha g_i > g_i,$ thus $Ky \not \in \Omega(y)$ even though $y \in \S$, contradicting the assumption that $\pi_0 \in \Pi$. We conclude that $u^* \not \in \bd(\Omega(x)).$ Since $u^* \in \Omega(x),$ $u^*$ must be an element of $\int(\Omega(x)).$
\hfill \qedsymbol

\subsection{Proof of Theorem \ref{thm:main}} \label{app:thm} \begin{enumerate}
    \item It suffices to show that the gauge map from $\B_\infty$ to $\hat{\Omega}_t$ is well-defined on $\int(\S).$ This is a direct result of Lemma \ref{lemma:intpoint}.
    \item By Lemmas \ref{lemma:intpoint} and \ref{lemma:polytope}, $\hat{\Omega}_t$ is a polytopic C-set. By \eqref{eqn:10-7-2}, $\gamma_{\hat{\Omega}_t}$ (and $\pi_\theta$) can be computed in closed form.
    \item A standard result from convex analysis shows that the subdifferential of \eqref{eqn:10-7-2} is defined for all $v \in \R^m$, for any polytopic C-set $\Q$ \cite{Blanchini2015}. If $\psi_\theta$ is a neural network, automatic differentiation techniques can be used to compute a subgradient of $\pi_\theta$ with respect to $\theta$ \cite{Gune2018}.
    \item This is due to the fact that $\psi_\theta$ is a universal function approximator for functions from $\S$ to $\B_\infty$ \cite{Hornik1989}. By \eqref{eqn:10-6-1} and Lemma \ref{lemma:bijection}, $\pi_\theta$ approximates any function in $\Pi$.
\end{enumerate}
    % A sufficiently large neural network $\psi_\theta$ with sigmoid or tanh activation in the output layer is a universal function approximator for functions from a bounded subset of $\R^n$ to $\B_\infty$. Choosing a point in $\B_\infty$ is equivalent to choosing a point in $\Omega(x_t)$ (by the other lemma).

\end{document}